\documentclass[journal,10pt]{IEEEtran}
\usepackage{color,algorithm,algorithmic,amsbsy,amsmath,amssymb,epsfig,bbm,mathrsfs,fancyhdr,fancyvrb,url}
\usepackage{graphicx}
\usepackage{amsmath}
\usepackage{ulem}
\usepackage{subcaption}
\usepackage{afterpage}  

\DeclareMathOperator{\erf}{erf}

\usepackage{xcolor,cite,etoolbox}

\newcommand{\beq}{\begin{equation}}
\newcommand{\eeq}{\end{equation}}
\newcommand{\beqn}{\begin{eqnarray}}
\newcommand{\eeqn}{\end{eqnarray}}

\usepackage{amsmath, amsthm, amssymb}
\usepackage[outdir=./]{epstopdf}
\usepackage{dsfont}
\usepackage{mathrsfs}

\newtheorem{Corollary}{Corollary}

\newtheorem{lemma}{Lemma}

\begin{document}
\title{Exact Coverage Analysis of Intelligent Reflecting Surfaces with Nakagami-{\em m} Channels}
\author{

\thanks{.)}
\thanks{Manuscript received XXX, XX, 2020; revised XXX, XX, 2020.}}

\markboth{IEEE Transactions on Vehicular Technology,~Vol.~XX, No.~XX, XXX~2020}
{}

\author{Hazem Ibrahim, Hina Tabassum,  and Uyen T. Nguyen
        \vspace{-10mm}
\IEEEcompsocitemizethanks{\IEEEcompsocthanksitem Hazem Ibrahim, Hina Tabassum, and Uyen T. Nguyen are with the Department
of Electrical Engineering and Computer Science, York University, Toronto, Ontario, M3J 1P3 Canada (e-mail:{hibrahim,hina,utn}@cse.yorku.ca).
}
}
\maketitle
\begin{abstract}
Intelligent Reflecting Surfaces (IRS) are a promising
solution to enhance the coverage of future wireless networks by
tuning low-cost passive reflecting elements (referred to as \textit{metasurfaces}),
thereby constructing a favorable wireless propagation
environment. Different from prior works, which assume
Rayleigh fading channels and do not consider the direct link
between a base station and a user, this article develops a framework based on moment
generation functions (MGF) to characterize
the coverage probability of a user in an IRS-aided wireless
systems with generic Nakagami-m fading channels 
in the presence of direct links. In addition, we demonstrate
that the proposed framework is tractable for both finite and
asymptotically large values of the metasurfaces. Furthermore,
we derive the channel hardening factor as a function of the
shape factor of Nakagami-m fading channel and the number of IRS
elements. Finally, we derive a closed-form expression to calculate
the maximum coverage range of the IRS for given network
parameters. Numerical results obtained from Monte-Carlo simulations validate the derived analytical results.
\end{abstract}
\IEEEdisplaynontitleabstractindextext
%
\IEEEpeerreviewmaketitle
\section{Introduction}\label{sec:introduction_report}
\IEEEPARstart{I}{ntelligent} reflecting surfaces (IRSs) are considered as among one of the key enabling technologies for 6G wireless networks. IRSs are fabricated surfaces of electromagnetic (EM) material that are electronically controlled with integrated electronics and have unique wireless communication capabilities. IRSs do not need any power supply, complex signal processing, or encoding and decoding processes, which are especially used to improve the signal quality at the receiver \cite{di2019smart,basar2019wireless}. IRSs  enable telecommunication operators to create a controllable channel propagation environment \cite{liaskos2018new}. 
Recent results have revealed that reconfigurable intelligent surfaces\footnote{IRSs are also referred to as reconfigurable intelligent surfaces.} can effectively control the wavefront, e.g., the phase, amplitude, frequency, and even polarization, of the impinging signals without complex signal processing operations.


Recently, a handful of research works have considered the coverage analysis of a user, assuming IRS-only transmissions with no direct transmission link  between a base station (BS) and a  user~\cite{yang2020coverage,zhang2020downlink,yue2020performance,basar2019wireless,Ref1,Ref2trigui2020comprehensive,Ref3ferreira2020bit,Ref4sharmay2020intelligent}. For instance, Ertugrul  et al. \cite{basar2019wireless} proposed a preliminary model to characterize the bit error rate of an IRS-assisted communication system  with large number of IRS elements and applied central limit theorem (CLT). Yang et al. \cite{yang2020coverage} characterized the coverage of an IRS-aided communication system with large number of IRS elements, applied CLT, and compared it with the relaying systems using Rayleigh fading channels.  Zhang et al.  \cite{zhang2020downlink} investigated the downlink performance of IRS-aided non-orthogonal multiple access (NOMA) networks via stochastic geometry. Yue and Liu \cite{yue2020performance} investigated the coverage probability in an IRS-assisted downlink NOMA networks.

{\color{black}Samuh and Salhab \cite{Ref1} approximated the outage probability, average symbol error probability (ASEP), and average channel capacity using the first term of Laguerre expansion.  Trigui et al. \cite{Ref2trigui2020comprehensive} derived the outage and ergodic capacity considering generalized  Fox’H fading channel. The expressions are in the form of multi-integral expressions and closed-forms are in the form of multi-variate Fox'H functions.  Considering Nakagami-m fading channels, Ferreira et al. \cite{Ref3ferreira2020bit} derived the exact bit error probability considering quadrature amplitude  and binary phase-shift keying (BPSK) modulations when the number of IRS elements, $n$, is equal to 2 and 3. The IRS-aided channel distributions for two and three IRS elements are given in the form of double and quadruple integrals, respectively. For large values of $n$,  central limit theorem (CLT) was applied. Sharma and Garg \cite{Ref4sharmay2020intelligent} investigated  the performance of IRSs with full duplex technology in terms of outage and error probabilities. The authors considered Nakagami-m fading channels and applied CLT.  }

None of the aforementioned research  considered direct links; therefore, it is difficult to identify scenarios in which  IRS-assisted transmissions outperform direct transmissions. Furthermore, most existing works assumed   asymptotically large numbers of IRS elements  and leverage on CLT \cite{lyu2020spatial,yue2020performance}. Recently,  Lyu and Zhang \cite{lyu2020spatial} 
 derived the spatial user throughput considering
Rayleigh fading, taking into account direct
links. However, to avoid tedious convolution of the probability density functions (PDFs) of the direct link and IRS-aided
link, they approximated the PDF of the cumulative channel gain using instead of with the gamma distribution. 

\textcolor{black}{Different from the aforementioned works, our contributions can be summarized as follows:
\begin{itemize}
    \item We provide a comprehensive moment generating function (MGF)-based framework to derive the exact coverage probability of a user in  an IRS-aided wireless system, assuming \textit{generic Nakagami-m} fading channels in the presence of \textit{direct link}.
    \item The proposed framework is analytically tractable for \textit{both finite and asymptotically large numbers of the IRS elements}, while allowing single integral coverage probability expressions. In contrast to CLT-based approximations or moment-based Gamma approximations, the proposed MGF-based approach can accurately reflect the behavior of IRS-assisted communication in all operating regimes. 
    \item  Furthermore, we derive the channel hardening factor as a function of the shape parameter of Nakagami-m fading channel and the number of IRS elements. The channel hardening factor reveals the conditions under which the  channel becomes nearly deterministic, resulting in improved reliability and  less frequent channel estimation. 
    \item We derive a \textit{closed-form expression to calculate the maximum coverage of the IRS} for the given network parameters.  
\item Numerical results obtained from Monte-Carlo simulations validate the analytical results and obtain useful insights related to the impact of a finite number of IRS elements, channel hardening in IRS-aided Nakagami-faded transmissions, the coverage range of the IRS, and the scenarios in which IRS-aided transmission gains can exceed direct transmission gains.
\end{itemize}    }




 \begin{figure}[t]
    \begin{center}
    \scalebox{0.45}[0.45]{\includegraphics{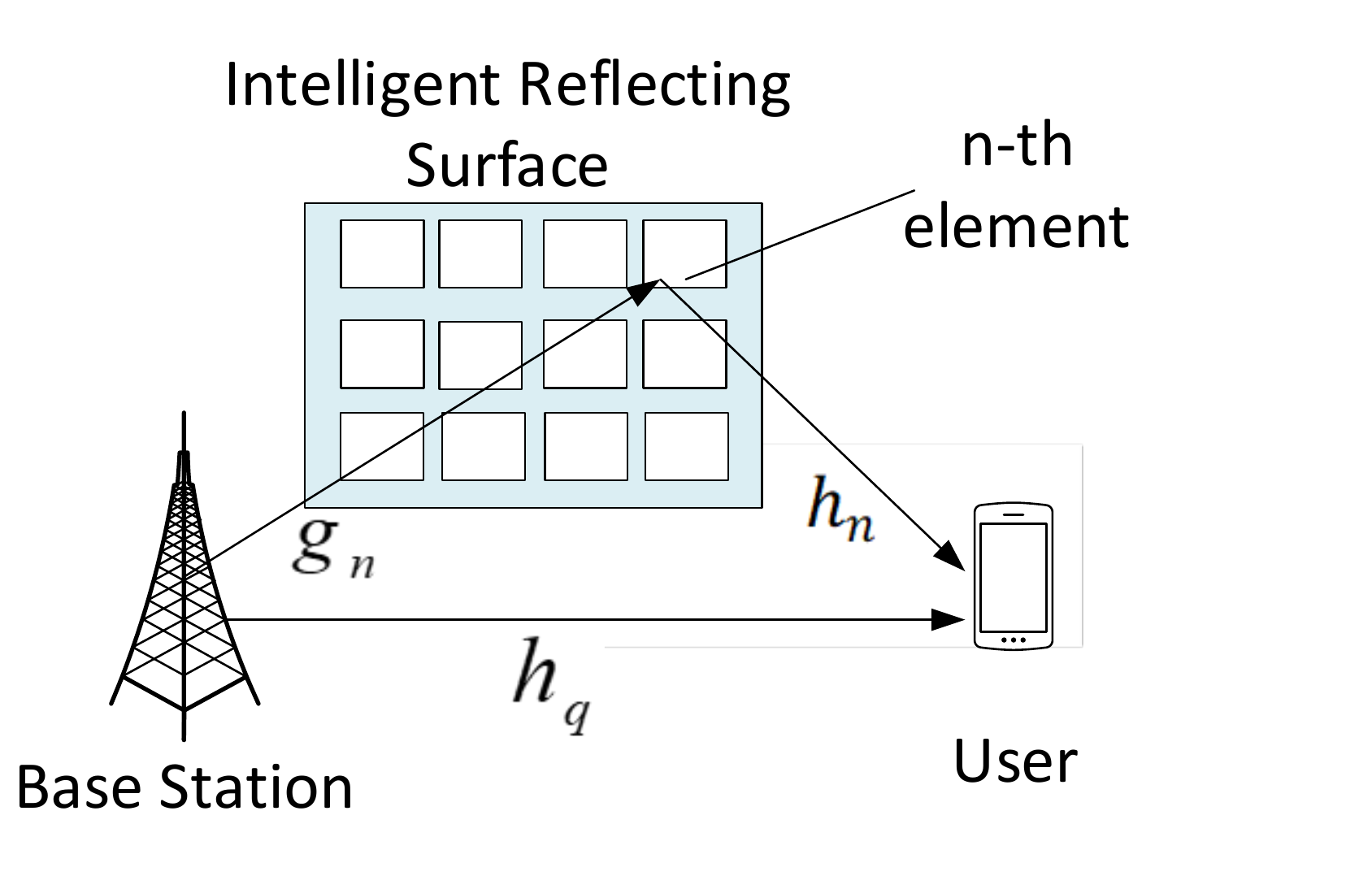}}
    \end{center}
    \caption{\small Intelligent reflecting surface-aided wireless system.}
     \label{IRS_arch}
    \end{figure}
\section{System Model and Assumptions}\label{model}

We consider a downlink IRS-aided communication system composed of a BS, a user, and an
IRS equipped with $N$ reflecting metasurfaces (as shown in Fig.~1). The BS and user are equipped with only one antenna. We assume that the IRS is placed between the BS and user  to improve the
user's signal quality in the presence of  direct transmissions. {\color{black}We assume that the IRS can obtain
the full channel state information (CSI)  to compute the phase shifts that maximize the signal-to-noise ratio (SNR)
at the receiver~\cite{wu2019towards,shafique2020optimization,basar2019wireless}.}

Let $\mathbf{g}\overset{\Delta}{=}[g_{1},..., g_{N}]^{T}\in \mathbb{C}^{N\times1}$, $\mathbf{h}\overset{\Delta}{=}[h_{1},..., h_{N}]^{T}\in \mathbb{C}^{N\times1}$, and $h_q\in \mathbb{C}$ denote the baseband equivalent channels  from the BS to the IRS, from the IRS to the user equipment (UE), and from the BS to the UE, respectively. Let $\mathbf{\Phi}\overset{\Delta}{=}\mathrm{diag}\{[e^{j\phi_{1}}, ..., e^{j\phi_{N}}]\}$ denote the phase-shifting matrix of the IRS, where $\phi_{n}\in [0,2\pi)$ is the phase shift by element $n$ on the incident signal, and $j$ denotes the imaginary unit. The \textit{cascaded BS-IRS-UE channel} is then modeled as a concatenation of three components, namely, the BS-IRS spatial stream link, IRS reflecting with phase shifts, and the IRS-UE link, and given by \cite{wu2019intelligent} as  $h_{c}\overset{\Delta}{=}\mathbf{g}^{T}\mathbf{\Phi}\mathbf{h}$. Note that the cascaded channel phase $\angle(g_{n}h_{n})$ for each IRS element $n =[1,... ,N]$ can be obtained via channel estimation. The IRS then adjusts the phase shift $\boldsymbol{\phi}=[\phi_{1},...,\phi_{N}]$ such that the $N$ reflected signals are of the same phase at its served UE's receiver by setting $\phi_{n} = -\angle(g_{n}h_{n})$, $n = 1, ... ,N$. Therefore, the overall cascaded channel gain is as follows:
\begin{align}
 |h_{c}|=|\mathbf{g}|^{T}|\mathbf{h}|=\sum_{n=1}^{N}|g_{n}||h_{n}|.   
\end{align}
We assume that the BS-UE channel phase $\angle h_q$ is also
known and the IRS can perform a common phase-shift such
that $h_c$ and $h_q$ are co-phased and hence coherently combined
at the UE [4], with the overall channel amplitude denoted by 
\begin{equation}
T=|h_c|+|h_q|.
\end{equation}
We denote $l$, $r$ and $d$ as the BS-UE, BS-IRS and UE-IRS distances, respectively. The channel amplitude from the BS to the $n$-th element of IRS and from the $n$-th element of the IRS to the UE can then be modeled, respectively, as follows:
\begin{align}
    |g_n|\overset{\Delta}{=}\zeta^{1/2}\varepsilon_{g} r_n^{-\alpha/2}, |h_n|&\overset{\Delta}{=}\zeta^{1/2}\varepsilon_{h} d_n^{-\alpha/2}.
\end{align}
 Similarly, the BS-UE channel amplitude is modeled as follows:
\begin{align}
    |h_q|\overset{\Delta}{=}\zeta \varepsilon_{q}l^{-\alpha/2},
\end{align}
where path-loss exponent $\alpha \geq 2$ and  $\zeta=(\frac{\text{carrier wavelength}}{4\pi})^{2}$ is  the near-field path loss factor at a reference distance of one meter (1 m) as a function of the carrier frequency.  The distribution of the Nakagami-m fading channel  gain $\varepsilon_{x}$ is thus given as follows:
\begin{equation}
   f_{\varepsilon_{x}}(x)=\frac{2m^{m}x^{2m-1}}{\Omega^{m}\Gamma(m)}\exp\left(\frac{-mx^{2}}{\Omega}\right), x>0, 
\end{equation}
 where $x\in\{g,h,q\}$ denotes the BS-IRS, IRS-UE and BS-UE links, respectively, and  $\Gamma(.)$ is the gamma function. The fading severity parameter is $m\in [1,2,...,\infty)$ and the mean fading power is denoted by $\Omega$. Note that 
$m\geq 0.5$ is the shape (or fading figure) parameter.   Rayleigh fading is a special case of Nakagami-{\em m} when  $m=1$. 

Let $P$ denote the BS downlink transmit power for each UE. If the typical UE  is served by the BS only, then its received SNR can be modeled as follows:
\begin{align}
    \mathrm{SNR}_{B}\overset{\Delta}{=} P|h_q|^{2}\hat{\sigma}^{-2}=P\zeta\varepsilon_{q}l^{-\alpha}\hat{\sigma}^{-2},
\end{align}
where $l=\sqrt{r^2+d^2-2rd\cos{({\color{black}\psi})}}$ and ${\color{black}\psi}$ is the angle between the BS-IRS and IRS-UE link, and $\hat{\sigma}^{2}$ is the variance of the additive white Gaussian noise at the UE receiver. On the other hand, if the UE is served by both the IRS and BS, then its received SNR can be modeled as follows:
\begin{align}
    \mathrm{SNR}_{S}&\overset{\Delta}{=} P(|h_c|+|h_q|)^{2}\hat{\sigma}^{-2}= P\left(\sum_{n=1}^{N}|g_{n}||h_{n}|+|h_q|\right)^{2}\hat{\sigma}^{-2}.
    \nonumber
\end{align}

\section{Exact Coverage Probability Analysis}

 The coverage probability is defined by the complementary cumulative distribution function (CCDF) of the SNR (i.e., $\mathbb{P}[{\rm SNR}>\theta]$, where $\theta$ denotes the predefined threshold for correct signal reception) and is characterized as follows:
\begin{align}
    \mathcal{C}_{S}&=\mathbb{P}[{\rm SNR_{S}}>\theta]
    =\mathbb{P}[P(|h_c|+|h_q|)^{2}\hat{\sigma}^{-2}>\theta]\notag,\\
    &={\color{black}\mathbb{P}\Big[T>\left(\frac{\theta \hat{\sigma}^{2}}{P}\right)^{1/2}\Big]=1-F_{T}\left(\left(\frac{\theta \hat{\sigma}^{2}}{P}\right)^{1/2}\right)},
\end{align}
{\color{black}where $F_{T}(.)$ is the  CDF of  $T$}. Following is our methodology: \textbf{(i)} deriving the Laplace transform (or moment generating function (MGF)) of $h_c$, $h_q$, and $T$ considering  asymptotically large values of reflecting elements $N$ in \textbf{Section~III.A} and finite values of $N$ in \textbf{Section~III.B}, \textbf{(ii)} deriving the characteristic function (CF) of $T$, and \textbf{(iii)} applying Gil-Pelaez inversion to the CF of $T$.
\subsection{Asymptotically Large Values of Reflecting Elements $N$}
The channel amplitude of the BS-IRS-UE signal that goes through element $n$ follows scaled double Nakagami-m, i.e.,
\begin{align}
    &|h_{c}|
    \stackrel{(a)}{\approx}\zeta r^{-\alpha/2}d^{-\alpha/2}\sum_{n=1}^{N} Y_n= \rho(r,d)\sum_{n=1}^{N} \varepsilon_g \varepsilon_h, \nonumber
\end{align}
{where (a) follows from the limited size of the UE and the IRS, i.e., $d_n\approx d$ and $r_n\approx r$.} Note that $Y_n$ denotes a double Nakagami-m distributed random variable with \textit{independent but not necessarily identically distributed (i.n.i.d)} $\varepsilon_h$ and $\varepsilon_g$ variables and $\rho(r,d)=\zeta r^{-\alpha/2}d^{-\alpha/2}$. 

For large values of $N$, the PDF of $|h_c|$ can be derived as:
\begin{align}\label{h_c}
    &{h_c}\stackrel{(a)}{\approx} \mathcal{N}\left(\rho(r,d)N \mathbb{E}[Y_n],\rho(r,d)^{2}N\mathrm{var}\{Y_n\}\right), 
\end{align}
where $\mathcal{N}(\mu,\sigma^2)$  denotes the Gaussian random variable with mean $\mu$ and variance $\sigma^2$. Note that step (a) follows from applying the central limit theorem (CLT) which dictates that $\sum_{n=1}^{N} Y_n$ approaches normal distribution as $N \rightarrow \infty$ and  scaling the normal distribution with
$\rho(r,d)$.  The distribution of a normal random variable $X$ scaled with a constant $c$  is given by $cX\sim
\mathcal{N}(c\mu,c^2\sigma^2)$, where $\mu$ is the mean and $\sigma^2$ is the variance. The $b$-th order moment of a double Nakagami-m RV $Y_n$ can be given as \cite{karagiannidis2005n}: 
\begin{align}
    \mathbb{E}[Y_n^{b}]=\prod_{i=1}^{2}\frac{\Gamma(m_{i}+b/2)}{\Gamma(m_{i})}\left(\frac{\Omega_{i}}{m_{i}}\right)^{b/2}.
\end{align}
Therefore, the mean and variance  of  $Y_n$ are as follows:
\begin{align}\notag
  &\mu_Y=\mathbb{E}[Y_n]\overset{\Delta}{=} \prod_{i=1}^{2}\frac{\Gamma(m_{i}+1/2)}{\Gamma(m_{i})}\left(\frac{\Omega_{i}}{m_{i}}\right)^{1/2},\\
   &\sigma_Y^2=\mathrm{var}\{Y_n\}\overset{\Delta}{=}\mathbb{E}[Y_n^{2}]-\mathbb{E}[Y_n]^{2}=\notag\\
   &=\prod_{i=1}^{2}\frac{\Gamma(m_{i}+1)}{\Gamma(m_{i})}\left(\frac{\Omega_{i}}{m_{i}}\right)-\left(\prod_{i=1}^{2}\frac{\Gamma(m_{i}+1/2)}{\Gamma(m_{i})}\left(\frac{\Omega_{i}}{m_{i}}\right)^{1/2}\right)^{2}.\notag
\end{align}
Now given the distribution of $|h_c|$, the coverage probability can be derived as in the following lemma.
\begin{lemma}
Applying the Gil-Pelaez inversion\cite{gil1951note}, the coverage probability can be obtained as
$
\mathcal{C}_S=1- F_{T}(t),
$
where
\begin{align}\label{Gil-pelaez-2hops}
F_{T}(t)=\frac{1}{2}+\frac{1}{\pi}\int_{0}^{\infty}\frac{\Im\left(e^{j\omega t}\varphi_{T}(\omega)\right)}{\omega}\text{d}\omega,
\end{align}
where $\Im(w)$ is imaginary part of $w\in \mathbb{C}$ and  $j\overset{\Delta}{=}\sqrt{-1}$. The characteristic function of $T$ can be given as follows:
\begin{equation}
    \varphi_{T}(\omega)=e^{-\mu_Y j \omega-\frac{1}{2}\sigma_Y^{2}\omega^{2}}\frac{\Gamma(m+0.5)}{\sqrt{\pi}}U\left(m,0.5,\frac{-(c_1 \omega)^{2}\Omega}{4m}\right).
    \notag
\end{equation}
\end{lemma}
\begin{proof}
Since $h_q$ is a scaled Nakagami-m distributed random variable, we first characterize the
 MGF of Nakagami-m variable $M_{Z}(s)$ as follows:
\begin{align}\label{mzs}
    M_{Z}(s)&=\int_{0}^{\infty}\frac{2m^{m}{Z}^{2m-1}}{\Omega^{m}\Gamma(m)}\exp\left(\frac{-m{Z}^{2}}{\Omega}\right)\exp(-sZ)dZ,\notag\\
    &=\frac{1}{\sqrt{\pi}}\Gamma(m+0.5)U\left(m,0.5,\frac{s^{2}\Omega}{4m}\right),
\end{align}
where $U(a,b,z)$ is the confluent hypergeometric function. On the other hand,  since we characterize  $h_c$ as normally distributed after applying CLT in Eq. (\ref{h_c}), the MGF of normally distributed variable $M_{h_c}(s)$ can be given as follows:
\begin{align}
    M_{h_c}(s)&=\exp\left(-\mu_Y s+\frac{1}{2}\sigma_Y^{2}s^{2}\right).
\end{align}
Since $|h_c|$ and $|h_q|$ are independent, the MGF of $T=|h_c|+|h_q|$ can be obtained by the multiplication of the MGFs of the random variables $h_c$ and $h_q$ as follows:
\begin{align}
&M_{T}(s){=}M_{h_{c}}(s) M_{h_q}(s)\stackrel{(a)}{=}M_{h_{c}}(s) M_{Z}(c_1 s),
\end{align}
\noindent where step (a) follows from the scaling property of MGF $M_{h_{q}}(c_1s)$ which is scaled with constant $c_1$, where $|h_q|\overset{\Delta}{=}\zeta^{1/2}l^{-\alpha/2}\varepsilon_{q}^{1/2}=c_1\varepsilon_{q}^{1/2}$ and $c_1\overset{\Delta}{=}\zeta^{1/2}l^{-\alpha/2}$.
\end{proof}
For Rayleigh fading channels, the MGF of $T$ can be  simplified as in the following corollary.
\begin{Corollary}
For $m=1$, the direct channel $h_q$ becomes scaled Rayleigh distributed random variable. Thus, using the identity $U(1,0.5,z)=e^{z}z^{0.5}\Gamma(-0.5,z)$, $M_{Z}(s)$ in \textbf{Lemma~1} can be simplified as follows:
\begin{equation}
    M_{Z}(s)
    =1-\frac{s\sqrt{\pi \Omega}}{2} e^{\frac{s^{2}\Omega}{4}} \mathrm{erfc}\left(\frac{s\sqrt{\Omega}}{2}\right),
\end{equation}
where {\rm erfc($\cdot$)} is the complementary error function.
\end{Corollary}
\subsection{Finite Values of Reflecting Elements $N$}
When $N$ is not a sufficiently large value, i.e., a finite value, the coverage probability
$\mathcal{C}_S=1- F_{T}(t)$ can be characterized via the following lemma.
\begin{lemma}
For finite values of $N$, we characterize the MGF of $ |h_{c}|=\zeta r^{-\alpha/2}d^{-\alpha/2}\sum_{n=1}^{N} Y_n$ considering that $Y_n$ is independent for all  $n$, i.e.,
\begin{equation}
    M_{h_c}(c_2 s)=\prod_{n=1}^{N} M_{Y_n}(c_2 s),
\end{equation}
where $c_2=\zeta r^{-\alpha/2}d^{-\alpha/2}$ is the scaling factor.  Since $Y_n$ is a double Nakagami-m RV, its MGF can be given as  follows \cite{karagiannidis2005n}:
\begin{align}\label{n_nakagami}
    M_{Y_n}(s)=\frac{1/\sqrt{\pi}}{\prod_{i=1}^{2}\Gamma(m_{i})}G_{2,2}^{2,2}\Bigg[\frac{4}{s^2}\prod_{i=1}^{2}\frac{m_i}{\Omega_{i}}\Bigg|\begin{matrix}
1/2 & 1  \\
m_{1} & m_2
\end{matrix}
\Bigg],
\end{align}
where $G[\cdot]$ is the Meijer's G-function \cite{jeffrey2007table}. 
\end{lemma}
 For independent and identically distributed (i.i.d.) $\varepsilon_g$ and $\varepsilon_h$, we simplify the MGF of $|h_{c}|$ as follows.
\begin{Corollary}
For i.i.d. $\varepsilon$, $Y_n$ follows a Gamma distribution, the MGF of $h_c$ can thus be given as  follows \cite{karagiannidis2006closed}:
\begin{align}\label{n_nakagami}
    M_{h_c}(c_2 s)=\left(1+\frac{c_2 s \Omega}{m}\right)^{-N m}.
\end{align}
\end{Corollary}
Since $|h_c|$ and $|h_q|$ are independent, the MGF of $T=|h_c|+|h_q|$ can be obtained as follows:
\begin{align}
&M_{T}(s){=}M_{h_{c}}(c_2 s) M_{h_q}(s)\stackrel{(a)}{=} M_{h_c}(c_2 s) M_{Z}(c_1 s),
\end{align}
where step (a) follows from  substituting $M_{h_{q}}(s)$ given in Eq. (\ref{mzs}).
The coverage probability can then be obtained using Gil-Pelaez inversion in \textbf{Lemma~1}.

\subsection{Channel Hardening for Double Nakagami-m Channels}
Let $\kappa$ denote the ratio between the mean $\mu_Y$ and standard deviation $\sigma_Y$ of $|h_c|$:
\small
\begin{align}
    \kappa&\overset{\Delta}{=}\frac{\mu_Y}{\sigma_Y}=\frac{\rho(r,d)N\mathbb{E}[Y]}{\sqrt{\rho(r,d)^{2}N \mathrm{var}\{Y\}}},\notag\\
    &\stackrel{(a)}{=}\frac{\sqrt{N}{[\Gamma(m+1/2)]^2}}{\sqrt{[\Gamma(m+1)]^2\Gamma(m)^2-[\Gamma(m+1/2)]^4}},
    \label{ratio_kaba}
\end{align}
\normalsize
where (a) follows when $m_1=m_2=m$. {The ratio $\kappa$ is a function of $N$ and $m$ and it is an indicator for channel hardening. Equation \eqref{ratio_kaba} implies a “channel hardening”
effect where the channel hardening increases as the number of IRS elements $N$ increases and fading severity $m$ increases. The impact of channel hardening increases with the increasing mean value of $Y$ and decreases with the increase in standard deviation of $Y$.  
}
When $m_1=m_2=m=1$, i.e., Rayleigh channels, we get $\kappa\overset{\Delta}
    =\frac{\sqrt{N}}{\sqrt{0.621}}$. When $m_1=m_2=m=0.5$, i.e.,  Nakagami-m channels with severe fading, the channel hardening decreases 
$\kappa=\frac{\sqrt{N}}{\sqrt{1.4674}}$.
The channel hardening effect can be used to find a condition where IRSs are beneficial to the coverage of typical user, that is, $\zeta l^{-\alpha}< \zeta^2 r^{-\alpha} d^{-\alpha}$. 

\subsection{IRS Coverage Range $D$}\label{D_rang}
To obtain insights related to the maximum coverage range of IRS transmissions, we derive the condition under which the outage probability becomes unity.
\begin{lemma}
The coverage range of IRS can be derived as:
\begin{align}
    d=\bigg[\frac{\sqrt{\frac{ \theta \hat{\sigma}^{2}}{2 P}}r^{\alpha/2}}{\zeta(N \mathbb{E}[Y_n])}\bigg]^{-2/\alpha}.
\end{align}
\end{lemma}
\begin{proof}
The outage probability is defined as  $\mathbb{P}[{\rm SNR}<\theta]$, where $\theta$ denotes the predefined threshold for correct signal reception and is characterized as follows:
\begin{align}
    \mathcal{O}&=\mathbb{P}[{\rm SNR}<\theta]
    =\mathbb{P}[P|h_c|^{2}\hat{\sigma}^{-2}<\theta]=F_{h_c}\left(\pm\sqrt{\frac{\theta \hat{\sigma}^{2}}{P}}\right),\notag
\end{align}
where $F_{h_c}(h)$ is the  CDF of  $|h_c|$. Since $|h_c|$ is a random variable with Gaussian distribution and its PDF is given in Eq. (\ref{h_c}), the outage can be given as  follows:
\begin{align}\label{h_c_1}
   \mathcal{O}&= 0.5\left(\erf\left(\frac{\sqrt{\frac{\theta \hat{\sigma}^{2}}{P}}-\mu}{\sigma\sqrt{2}}\right) -\erf\left(\frac{-\sqrt{\frac{\theta \hat{\sigma}^{2}}{P}}-\mu}{\sigma\sqrt{2}}\right)\right), 
\end{align}
where $\mu=\rho(r,d)N \mathbb{E}[Y_n]$, $\sigma=\sqrt{\rho(r,d)^{2}N\mathrm{var}\{Y_n\}}$, and  $\rho(r,d)=\zeta r^{-\alpha/2}d^{-\alpha/2}$. To determine the coverage range of the IRS, we calculate $d$ that satisfies  $\mathcal{O}= F_{h_c}(h)=1$ as:
\begin{align}\label{app0}
  \erf\left(\frac{\sqrt{\frac{\theta \hat{\sigma}^{2}}{P}}-\mu}{\sigma\sqrt{2}}\right) +\erf\left(\frac{\sqrt{\frac{\theta \hat{\sigma}^{2}}{P}}+\mu}{\sigma\sqrt{2}}\right)=2.
\end{align}
The aforementioned equation holds closely when
\begin{equation}\label{app1}
    \sqrt{\frac{\theta \hat{\sigma}^{2}}{P}} =4 \mu \quad \text{and} \quad \frac{\mu}{\sigma \sqrt{2}}=\frac{\sqrt{N}\mathbb{E}[Y_n]}{\sqrt{\mathrm{var}[Y_n]}} \geq 0.5.
\end{equation}
Finally, solving \eqref{app1} results in the coverage range as given in \textbf{Lemma~3}.
\end{proof}

\section{Numerical Results and Discussions}\label{results_mm}
In this section, we first  present the simulation parameters. Then, we validate our numerical results using Monte-Carlo simulations and use the developed analytical models to obtain insights related to the coverage probability of the typical user as a function of the number of IRS elements, SNR threshold, fading severity, and distance between the UE and the IRS.

\subsubsection{Simulation Parameters}\label{parameter_sec}
Unless otherwise stated, we use the following simulation parameters throughout our numerical results. The transmission power of the BS  is $P=2.5$ Watts. The distance $l=500$ meters when {\color{black}$\psi=85^\circ$}, the distance $r=500$ meters, and the distance $d=100$ meters. 
The path loss exponent for the BS is set to $\alpha=4$. The network downlink bandwidth is $W=100$ MHz allocated for the BS. The receiver noise is calculated as \cite{ibrahim2019meta,iibrahim2019meta}   $\hat{\sigma}^{2}=-174\text{ dBm/Hz}+10\log_{10}(W)+10\text{ dB}$. 

\subsubsection{Coverage as a function of SNR threshold}\label{validation}
 The proposed analytical model is validated using Monte Carlo simulations implemented in MATLAB. We performed  simulations over $10,000$ network configurations with different Nakagami-m fading severity parameters and IRS equipped with asymptotically large number of IRS elements. As can be seen in Fig.~\ref{coverage_prob_ver}, the analytical results of the derived coverage probability under different Nakagami fading parameters $m=0.5, 1, 2$ match the simulation results. This confirms the accuracy of the analytical expressions derived above for our model. As observed from Fig. \ref{coverage_prob_ver}, as $m$ increases, the coverage probability increases due to decreasing fading severity (increasing m). When $m=1$, this represents the coverage probability of Rayleigh fading. 
 \begin{figure}[!h]
    \begin{center}
    \scalebox{0.55}[0.55]{\includegraphics{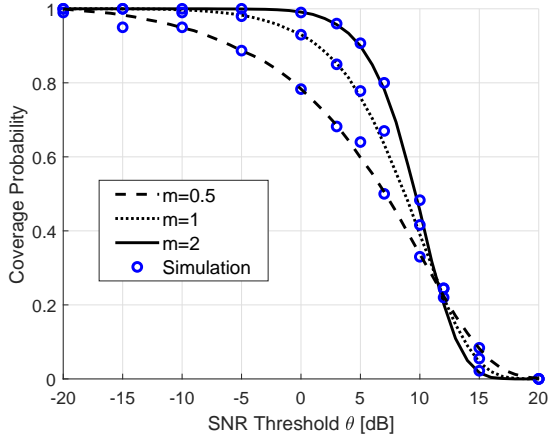}}
    \end{center}
    \caption{\small Coverage probability as a function of SNR threshold $\theta$ for $N=500$ reflecting antenna elements, $\alpha=4$, $\theta=5$ dB, $d=100$ m, $r=500$ m, $P=2.5$ W.}
     \label{coverage_prob_ver}
    \end{figure}
\subsubsection{Coverage as a function of number of IRS Elements}\label{cov_disc}
Fig. \ref{coverage_m} illustrates the coverage probability $\mathcal{C}_{S}$ of the typical user as a function of finite number of IRS elements and distance from the IRS.  We can see that by increasing the distance from the IRS and lower number of IRS elements, the coverage of the typical user significantly decreases. The IRS link is beneficial for small values of distance $d$ between the user and the IRS. In particular, for large values of $d$, the coverage from IRS link is nearly zero and the achieved 60\% coverage can be observed only due to the direct link which is independent of the distance between the UE and the IRS. 

{\color{black} Most of the existing  research works considered asymptotically large numbers of IRS elements and relied on CLT, which may not yield accurate results. The reason is that an IRS surface is typically of a finite size and thereby possesses a finite number of IRS elements in practice. Fig.~3  offers insights related to the impact of a finite number of IRS elements on the coverage probability. This numerical result quantifies the performance gap considering both finite and infinite number of IRS elements (as illustrated by the difference in the coverage probability when $N=10$ and when $N=500$) and signifies the importance of exact coverage analysis. }
    \begin{figure}[!h]
    \begin{center}
    \scalebox{0.55}[0.55]{\includegraphics{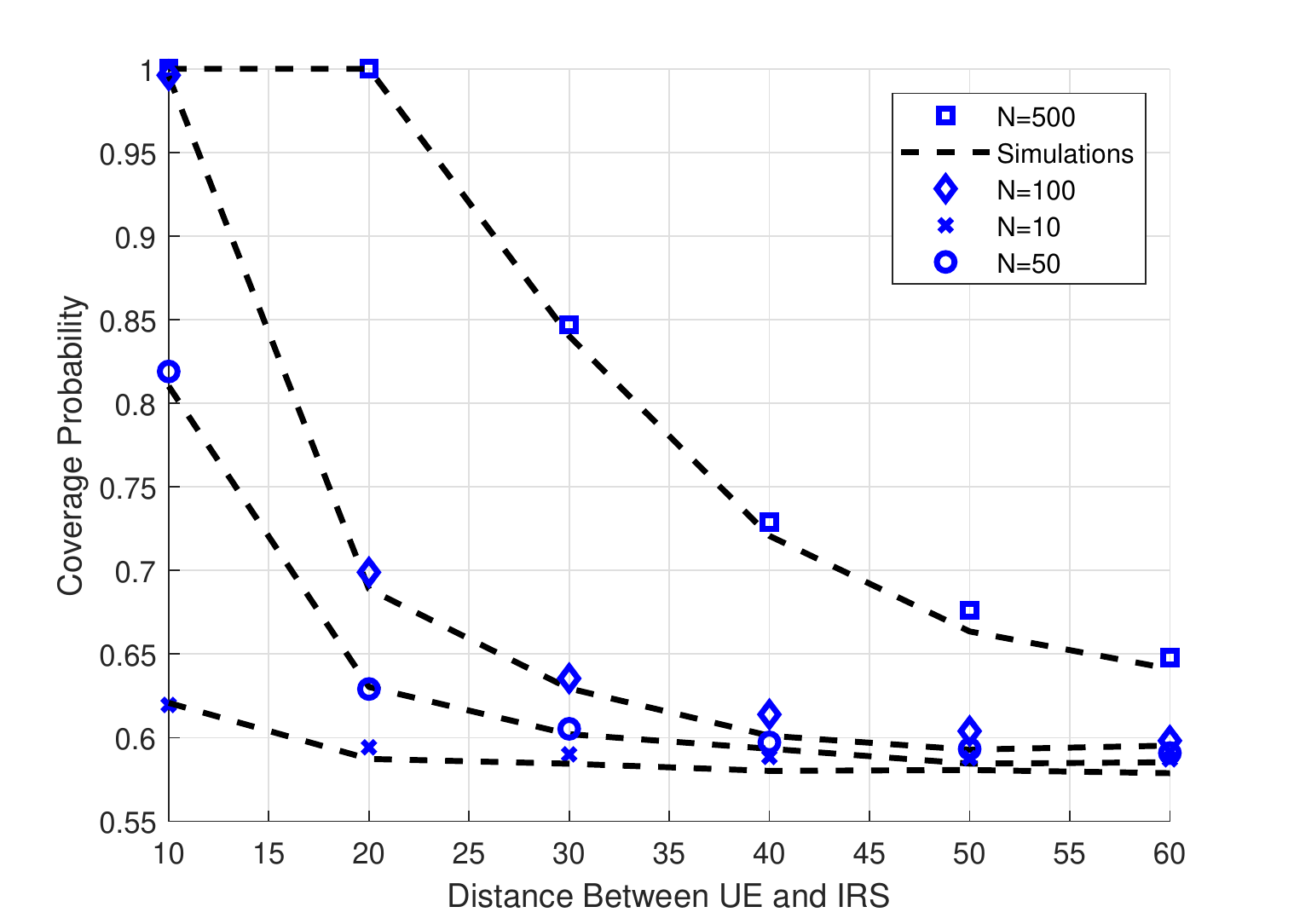}}
    \end{center}
    \caption{\small Coverage probability as a function of the IRS elements $N$ when $m=0.5$, $\alpha=4$, $\theta=5$dB, $r=500$ m, $p=2.5$ W (The blue data points are obtained from analysis. The dotted black lines are obtained from simulations).}
     \label{coverage_m}
    \end{figure}
    \subsubsection{Coverage probability as a function of the user's distance from the IRS}
Fig. \ref{1m} illustrates the coverage probability as a function of the user's distance from the IRS and offers insights related to the benefits of IRS-aided transmissions over the direct link. {\color{black}In particular, the graph provides a comparative analysis of three modes of operation, 1) direct transmission, 2) IRS-aided transmission, and 3) IRS-aided transmission with direct link.  We highlight points A and B in Fig. \ref{1m} where operations modes can be switched in order to enhance coverage probability.  Furthermore, we note that the IRS link is beneficial for small values of distance $d$ between the user and the IRS. The graph in Fig. \ref{1m} indicates the maximum coverage distance of the IRS.} That is, in the absence of a direct link, a user can choose the IRS if its distance is less than 30 m from the IRS equipped with 500 elements. If the direct link is available, it is beneficial to combine both signals up to a distance of 70 m. On the other hand, if the IRS is equipped with 100 elements, the user can choose the IRS if its distance is less than 20 m. If the direct link is available, it is beneficial to combine both signals up to a distance of 40~m. 
\textcolor{black}{Our closed-form result given in Lemma~3 and numerical result in Fig.~4 provide visualization of the maximum coverage range of a given IRS and thus can be a useful tool for interference evaluation in large-scale networks.  }
      \begin{figure}[!h]
    \begin{center}
    \scalebox{0.55}[0.55]{\includegraphics{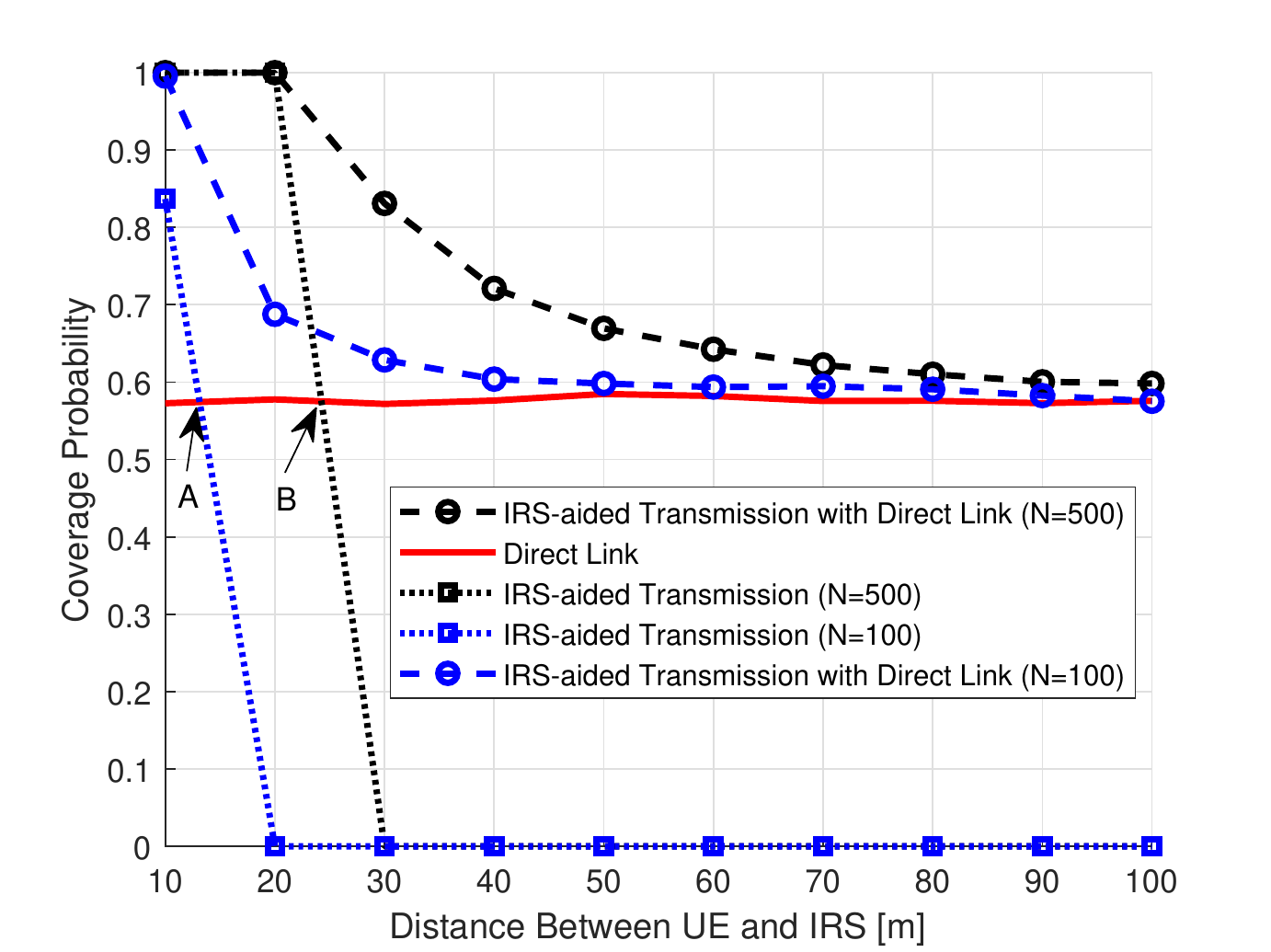}}
    \end{center}
    \caption{\small Comparison of the coverage probability as a function of the user's distance from the IRS, $m=0.5$.}
     \label{1m}
    \end{figure}
\subsubsection{Channel hardening}
In order to show the effect of channel hardening, in Fig. \ref{coverage_N_2}, we plot $\kappa$ as a function of the number of reflecting antenna elements $N$. As we can see in Fig. \ref{coverage_N_2}, given a value $N$, as $m$ increases (fading  decreases), the channel hardening factor $\kappa$  also increases. Channel hardening correlates with
a nearly deterministic channel with improved reliability, and requires less frequent channel estimation.
     \begin{figure}[!h]
    \begin{center}
    \scalebox{0.55}[0.55]{\includegraphics{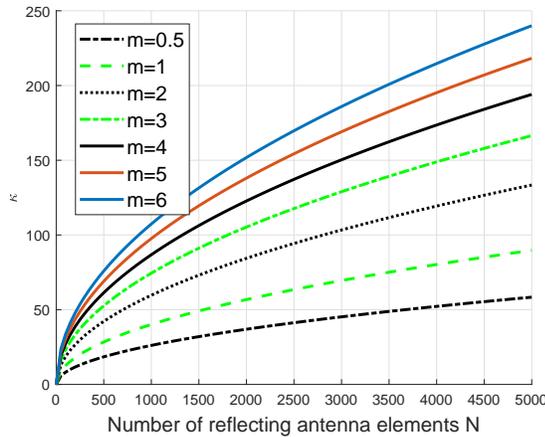}}
    \end{center}
    \caption{\small $\kappa$ as a function of number of reflecting antenna elements $N$ for different values of $m$.}
     \label{coverage_N_2}
    \end{figure}

\section{Conclusion}\label{conc_future_mm}
We characterize the coverage probability of IRS-aided communication networks with  Nakagami-m channels, taking into account the direct link between the BS and UE in our model. The results reveal that the number of intelligent reflecting surfaces $N$ has a significant impact on the system performance, and that the use of IRSs enhances coverage for edge users. {\color{black}
Our results show that the assumption of asymptotically large  numbers of IRS in most of existing work can overestimate the users’ coverage probability compared  to  the  exact  users’  coverage  probability. We provide a comparative analysis of three important IRS system modes of  operation, 1) direct transmission, 2) IRS-aided transmission, and 3) IRS-aided transmission with direct link. We  highlight  points  where modes have to be switched to achieve better coverage probability.} Moreover, increasing the number of intelligent reflecting surfaces $N$ and reducing the fading severity (i.e., increasing the fading severity parameter) enhances the channel hardening in IRS-aided communication networks. Finally, we propose a closed-form expression  to  characterize  the  IRS coverage range for any network parameters. 

\bibliographystyle{IEEEtran}
\bibliography{IoT-massive}

\begin{thebibliography}{10}
\providecommand{\url}[1]{#1}
\csname url@samestyle\endcsname
\providecommand{\newblock}{\relax}
\providecommand{\bibinfo}[2]{#2}
\providecommand{\BIBentrySTDinterwordspacing}{\spaceskip=0pt\relax}
\providecommand{\BIBentryALTinterwordstretchfactor}{4}
\providecommand{\BIBentryALTinterwordspacing}{\spaceskip=\fontdimen2\font plus
\BIBentryALTinterwordstretchfactor\fontdimen3\font minus
  \fontdimen4\font\relax}
\providecommand{\BIBforeignlanguage}[2]{{%
\expandafter\ifx\csname l@#1\endcsname\relax
\typeout{** WARNING: IEEEtran.bst: No hyphenation pattern has been}%
\typeout{** loaded for the language `#1'. Using the pattern for}%
\typeout{** the default language instead.}%
\else
\language=\csname l@#1\endcsname
\fi
#2}}
\providecommand{\BIBdecl}{\relax}
\BIBdecl

\bibitem{di2019smart}
M.~Di~Renzo, M.~Debbah, D.-T. Phan-Huy, A.~Zappone, M.-S. Alouini, C.~Yuen,
  V.~Sciancalepore, G.~C. Alexandropoulos, J.~Hoydis, H.~Gacanin \emph{et~al.},
  ``Smart radio environments empowered by reconfigurable {AI} meta-surfaces: An
  idea whose time has come,'' \emph{EURASIP Journal on Wireless Communications
  and Networking}, vol. 2019, no.~1, pp. 1--20, 2019.

\bibitem{basar2019wireless}
E.~Basar, M.~Di~Renzo, J.~De~Rosny, M.~Debbah, M.-S. Alouini, and R.~Zhang,
  ``Wireless communications through reconfigurable intelligent surfaces,''
  \emph{IEEE Access}, vol.~7, pp. 116\,753--116\,773, 2019.

\bibitem{liaskos2018new}
C.~Liaskos, S.~Nie, A.~Tsioliaridou, A.~Pitsillides, S.~Ioannidis, and
  I.~Akyildiz, ``A new wireless communication paradigm through
  software-controlled metasurfaces,'' \emph{IEEE Communications Magazine},
  vol.~56, no.~9, pp. 162--169, 2018.

\bibitem{yang2020coverage}
L.~Yang, Y.~Yang, M.~O. Hasna, and M.-S. Alouini, ``Coverage, probability of
  {SNR} gain, and {DOR} analysis of {RIS}-aided communication systems,''
  \emph{IEEE Wireless Communications Letters}, 2020.

\bibitem{zhang2020downlink}
C.~Zhang, W.~Yi, Y.~Liu, Z.~Qin, and K.~K. Chai, ``Downlink analysis for
  reconfigurable intelligent surfaces aided {NOMA} networks,'' \emph{arXiv
  preprint arXiv:2006.13260}, 2020.

\bibitem{yue2020performance}
X.~Yue and Y.~Liu, ``Performance analysis of intelligent reflecting surface
  assisted {NOMA} networks,'' \emph{arXiv preprint arXiv:2002.09907}, 2020.

\bibitem{Ref1}
M.~H. Samuh and A.~M. Salhab, ``Performance analysis of reconfigurable
  intelligent surfaces over {N}akagami-m fading channels,'' \emph{arXiv
  preprint arXiv:2010.07841}, 2020.

\bibitem{Ref2trigui2020comprehensive}
I.~Trigui, W.~Ajib, and W.-P. Zhu, ``A comprehensive study of reconfigurable
  intelligent surfaces in generalized fading,'' \emph{arXiv preprint
  arXiv:2004.02922}, 2020.

\bibitem{Ref3ferreira2020bit}
R.~C. Ferreira, M.~S. Facina, F.~A. De~Figueiredo, G.~Fraidenraich, and E.~R.
  De~Lima, ``Bit error probability for large intelligent surfaces under
  double-nakagami fading channels,'' \emph{IEEE Open Journal of the
  Communications Society}, vol.~1, pp. 750--759, 2020.

\bibitem{Ref4sharmay2020intelligent}
P.~K. Sharmay and P.~Garg, ``Intelligent reflecting surfaces to achieve the
  full-duplex wireless communication,'' \emph{IEEE Communications Letters},
  2020.

\bibitem{lyu2020spatial}
J.~Lyu and R.~Zhang, ``Spatial throughput characterization for intelligent
  reflecting surface aided multiuser system,'' \emph{IEEE Wireless
  Communications Letters}, 2020.

\bibitem{wu2019towards}
Q.~Wu and R.~Zhang, ``Towards smart and reconfigurable environment:
  {I}ntelligent reflecting surface aided wireless network,'' \emph{IEEE
  Communications Magazine}, vol.~58, no.~1, pp. 106--112, 2019.

\bibitem{shafique2020optimization}
T.~Shafique, H.~Tabassum, and E.~Hossain, ``Optimization of wireless relaying
  with flexible {UAV}-borne reflecting surfaces,'' \emph{IEEE Transactions on
  Communications}, 2020.

\bibitem{wu2019intelligent}
Q.~Wu and R.~Zhang, ``Intelligent reflecting surface enhanced wireless network
  via joint active and passive beamforming,'' \emph{IEEE Transactions on
  Wireless Communications}, vol.~18, no.~11, pp. 5394--5409, 2019.

\bibitem{karagiannidis2005n}
G.~K. Karagiannidis, N.~C. Sagias, and T.~Mathiopoulos, ``The {N}* nakagami
  fading channel model,'' in \emph{2005 2nd International Symposium on Wireless
  Communication Systems}.\hskip 1em plus 0.5em minus 0.4em\relax IEEE, 2005,
  pp. 185--189.

\bibitem{gil1951note}
J.~Gil-Pelaez, ``Note on the inversion theorem,'' \emph{Biometrika}, vol.~38,
  no. 3-4, pp. 481--482, 1951.

\bibitem{jeffrey2007table}
A.~Jeffrey and D.~Zwillinger, \emph{Table of integrals, series, and
  products}.\hskip 1em plus 0.5em minus 0.4em\relax Elsevier, 2007.

\bibitem{karagiannidis2006closed}
G.~K. Karagiannidis, N.~C. Sagias, and T.~A. Tsiftsis, ``Closed-form statistics
  for the sum of squared nakagami-m variates and its applications,'' \emph{IEEE
  Transactions on Communications}, vol.~54, no.~8, pp. 1353--1359, 2006.

\bibitem{ibrahim2019meta}
H.~Ibrahim, H.~Tabassum, and U.~T. Nguyen, ``The meta distributions of the
  {SIR}/{SNR} and data rate in coexisting sub-6{GH}z and millimeter-wave
  cellular networks,'' \emph{IEEE Open Journal of the Communications Society},
  vol.~1, pp. 1213--1229, 2020.

\bibitem{iibrahim2019meta}
------, ``Meta distribution of {SIR} in dual-hop {I}nternet-of-{T}hings
  ({I}o{T}) networks,'' in \emph{ICC 2019-2019 IEEE International Conference on
  Communications (ICC)}.\hskip 1em plus 0.5em minus 0.4em\relax IEEE, 2019, pp.
  1--7.

\end{thebibliography}
%
%

\end{document}